\documentclass[ninept]{article}
\usepackage{spconf,amsmath,amsfonts, graphicx,hyperref}
\usepackage{amsthm}
\usepackage{tikz}

\theoremstyle{definition}
\newtheorem{lemma}{Lemma}

\theoremstyle{definition}

\usepackage{color}

\usepackage{soul}

\title{Hybrid channel estimation with Quantized Phase Feedback for Over-the-Air Computation}
\name{Martin Dahl and Erik G. Larsson
	\thanks{This work was supported in part by ELLIIT, the Swedish Research Council (VR), and the Knut and Alice Wallenberg (KAW) Foundation. Martin Dahl is affiliated with the Wallenberg Autonomous Systems Program (WASP) Graduate School.}}
\address{Department of Electrical Engineering (ISY), Linköping University, 581 83 Linköping, Sweden}
\ninept

\begin{document}
%
\maketitle
\begin{abstract}
To reduce the signaling overhead of over-the-air computation, a hybrid channel estimation scheme is proposed, where reciprocity-based and feedback-based channel estimation are combined. In particular, the impact of quantized phase-feedback is studied while the amplitude is assumed estimated exactly. The scheme enables selecting the estimation precision of amplitude and phase separately, depending on the importance of each. Two variants of the scheme are proposed: As shown through simulations and theory, the second variant with reciprocity-based estimation of the channel phase, and optimal quantization of phase feedback, can outperform the first variant estimating the phase by feedback only. 
\end{abstract}
\begin{keywords}
Over-the-air computation, channel estimation, quantization, reciprocity calibration, feedback
\end{keywords}
\section{Introduction}
\label{sec:intro}
Over-the-Air computation (OAC), closely related to distributed beamforming \cite{nanzer2021distributed, sahin2025distributed_phase}, has been proposed as a method to aggregate distributed values over a wireless network \cite{csahin2023survey}. Traditionally, the devices may transmit their values separate in time, frequency and space without interference to an access-point (AP), where the individual values can be decoded and the sum computed. Conversely, in OAC all devices transmit simultaneously on the same frequency and spatial resource, exploiting superposition to compute the sum. It follows that the main applications of OAC are where the interest is the sum of values from distributed devices, rather than the individual values. The most promising such application is federated learning \cite{csahin2023survey, kairouz2021advances}. 

Theoretically, OAC has a resource use of $\mathcal{O}(1)$ compared to the traditional $\mathcal{O}(K)$ for $K$ devices in the network. However, while OAC appears more efficient, repeated estimation of uplink (UL) channel state information (CSI) at all devices is necessary because of mobility and phase drift of the local oscillators. The CSI, in particular its phase, is required for the transmitted values to add up coherently at the AP. Previous work commonly assume CSI is known exactly \cite{zhu2019broadband}, but also investigate OAC with inexact CSI knowledge \cite{evgenidis2023over}. Additionally, a non-coherent OAC not requiring channel phase has been studied \cite{lee2024performance}, relying on the second moment of the channel amplitude, which varies slower than the instantaneous phase making it easier to estimate. For coherent OAC, the focus herein, the UL CSI can in general be estimated using one of two options:

The first option is to have the devices transmit UL pilots, to which the AP responds with feedback. This method is simple and straightforward, however, the communication cost scales linearly with the number of devices counteracting the OAC benefit of reducing such linear growth. This approach was studied for OAC in \cite{xie2023optimal} and \cite{ang2019robust}.

A second option is to calibrate\footnote{As noted by \cite{larsson2024massive}, alternative terms in the literature for "calibration" of phase are "synchronization" and "alignment".} the devices in phase for joint reciprocity through bi-directional signaling between each device and the AP. Such calibration enables estimating the UL CSI by a single broadcast downlink (DL) pilot from the AP \cite{shepard2012argos}. A similar setting is distributed multiple-input multiple-output (MIMO) where APs may calibrate over-the-air for joint beamforming to a user terminal \cite{ganesan2023beamsync}: In OAC the APs are exchanged for the OAC devices and the user terminal for the OAC AP.\footnote{In distributed MIMO there may also be a wired fronthaul distributing a common calibration signal, which is not the case for the devices in OAC.} Channel estimation with reciprocity has been considered in previous work on OAC, with \cite{guo2021over} and without calibration \cite{xie2023optimal}. However, in practice phase calibration is necessary for reciprocity based channel estimation: While wireless channels (antenna-antenna) are reciprocal, the effect of device and AP hardware makes the total channel (hardware-hardware) non-reciprocal  \cite{nissel2022correctly}. This is caused by unknown time-varying phase- and amplitude-factors in the hardware \cite{larsson2024massive, shepard2012argos}. The main time-varying factor is a drift in phase of the local oscillators \cite{nissel2022correctly}, primarily caused by carrier frequency offset \cite{sahin2025distributed_phase} and phase noise \cite{piemontese2024discrete}. Consequently, the devices must periodically re-calibrate to maintain the joint reciprocity \cite{dahl2024over}

\textbf{Contribution: }A hybrid channel estimation scheme for coherent OAC is proposed, exploiting both reciprocity and feedback based channel estimation with quantization of the phase feedback.
Two variants of the scheme are studied, where analytically derived OAC performance is provided for the first variant, and simulated performance for the second. Much work exists on quantization of phase \cite{smith2003comparison, dowhuszko2014performance} outside of the OAC context, as well as for OAC with joint quantization of amplitude and phase \cite{tsinos2023over}. The novelty lies in exploiting reciprocity \textit{and} feedback, as well as separate quantization of the phase feedback for OAC. Apart from \cite{lee2024performance} and our own previous work \cite{dahl2024over} which only considered reciprocity based channel estimation without quantization, the isolated effect of phase errors on OAC has not been studied before.

\section{System Objective and Model}
\label{sec:format}

Consider $K$ devices, each with a value $v_k\in\mathbb{C}$ such that $\mathbb{E}[v_k]=0$, $\mathbb{E}\left[|v_k|^2\right]=1$, i.i.d. $\forall k$. The system objective is to compute an estimate $\widehat{v}$ of the sum
\begin{equation}
	v \equiv \sum_{k=1}^{K}v_k,
\end{equation}   
minimizing the mean-square error (MSE)
\begin{equation}\label{eq:MSE_goal}
	\mathbb{E}\left[|v-\widehat{v}|^2\right].
\end{equation}
To this end, each device transmits a symbol $x_k=a_kv_k\in\mathbb{C}$ subject to an average power constraint $\mathbb{E}\left[|x_k|^2\right]=P>0$, where $a_k\in\mathbb{C}$ is a precoding coefficient and $P$ the maximum mean power. Then, the AP measures
\begin{equation}\label{eq:receive_OAC}
	y = \sum_{k=1}^{K}g_kx_k + n\in\mathbb{C},
\end{equation}
where $g_k\in\mathbb{C}, \text{ }\forall k$ is the CSI of the UL channel and ${n\in\mathbb{C}}$ thermal noise. A time-index is omitted for notational ease $\left(y\equiv y^t\right)$, however, the channel is block-fading and admits new i.i.d. instances of $g_k$ and $n$ for $t'\neq t$. The estimator $\widehat{v}$ is heuristically defined as
\begin{equation}
	\widehat{v} = by,
\end{equation}
where $b\in\mathbb{C}$ is the receive coefficient.   

\section{Precoding and Power Control}
\label{sec:pagestyle}
To select $a_k$ and $b$, the ``static-channel'' power-control scheme from \cite{cao2020optimized} is applied with modified assumptions: Consider ${g_k = |g_k|e^{j\phi_k}}$, where $|g_k|$ is the channel amplitude and $\phi_k\equiv \angle g_k$ the channel phase. Herein, the devices obtain 
channel estimates as
\begin{equation}
	\widehat{g_k}\equiv |g_k| e^{j\widehat{\phi_k}},
\end{equation}
where the amplitude of the channel is assumed to be estimated exactly $\left(|\widehat{g_k}|=|g_k|\right)$, but not the phase $(\widehat{\phi_k}\neq \phi_k)$, assuming a high AP power budget ($\gg P$) and stable hardware amplitude-factors. In contrast, \cite{cao2020optimized} assumes both $|g_k|$ and $\phi_k$ to be estimated exactly, which results in their power control scheme being agnostic to phase errors.  Furthermore,  \cite{cao2020optimized} only applies channel inversion at device $k$ if $|g_k|$ is sufficiently large, selecting $a_k$ and $b$ as follows
\begin{equation}\label{eq:power_control_coeff}
	a_k = \frac{e^{-j\widehat{\phi_k}}}{|g_k|}, b = 1, \forall k.
\end{equation}
To focus on phase errors, we assume a high channel power such that channel inversion is very likely, resulting approximately in the following estimate
\begin{equation}\label{eq:OAC_estimate}
	\widehat{v} = \sum_{k=1}^{K}v_ke^{j\left(\phi_k-\widehat{\phi_k}\right)} + n.
\end{equation}
Note that with (\ref{eq:power_control_coeff}) $\widehat{v}$ is not necessarily unbiased conditioned on $v$, but minimizes its MSE (\ref{eq:MSE_goal}) in the case of exact channel phase knowledge. 

\section{Quantization, Calibration and phase noise}
The following sections give a review on important concepts to understand the proposed scheme.

\subsection{Uniform and Lloyd-Max Quantization}
Let $Q(x)$ be an $N$-bit quantizer with $2^N$ quantization levels $q\in\Phi, (|\Phi|=2^N)$ as follows
\begin{equation}
	Q(x) = \underset{q\in\Phi}{\text{ arg min }} |q-x|.
\end{equation}
For the \textit{uniform} quantizer (UQ) the levels are selected as
\begin{equation}\label{eq:UQ_levels}
	\Phi = \left\{\frac{i\pi}{2^{N-1}}\bigg | i\in\left\{0,\dots,2^N-1\right\}\right\},
\end{equation}
for angles evenly spread on the unit-circle, starting at $0$\footnote{For random angles on the unit circle, the mean square quantization error is equal regardless of where you start, but $0$ is intuitive for the problem.}. For the \textit{Lloyd-Max} quantizer (LMQ) \cite{lloyd1982least}, parameterized by $\alpha$, the levels are as follows\footnote{The zero-mean, $\alpha$-variance normal distribution is denoted by $\mathcal{N}(0,\alpha)$, $\mathcal{CN}(0,\alpha)\equiv\mathcal{N}(0,\alpha/2) + j\mathcal{N}(0,\alpha/2)$ is the complex normal.}
\begin{equation}
	\Phi_\alpha = \left\{q_i\bigg| \underset{\{q_i\}}{\text{ min }}\mathbb{E}_{x\sim\mathcal{N}(0,\alpha)}\left[\underset{q\in\{q_i\}}{\text{ min }}(x-q)^2\right]\right\}.
\end{equation}
That is, the levels of the LMQ are selected to minimize the mean square quantization error for a given distribution, in this case for Gaussian data. The LMQ levels are generally only available as numerical approximates, visualized in Figure \ref{fig:quantizers_unit_circle} together with levels for the uniform quantizer when $N=1,2$. This quantizer only makes sense for low $\alpha$ $(\leq 1)$, as it does not consider the phase $\text{mod }\pi$. 

\subsection{Channel Estimation with Calibrated Reciprocity}\label{sec:reciprocity_calibration}
As discussed, wireless channels are not reciprocal in practice. The UL channel $g_{k\rightarrow \text{AP}}\equiv g_k$ from device $k$ to the AP is generally not equal to the DL channel $g_{\text{AP}\rightarrow k}$. This can be modeled with multiplicative coefficients $t_k, r_k, t_\text{AP}, r_\text{AP}\in\mathbb{C}$, giving the total UL and DL channels $g_{k\rightarrow \text{AP}} = t_kh_kr_\text{AP}$ and $\text{ }g_{\text{AP}\rightarrow k} = t_\text{AP}h_kr_k$
where $h_k\in\mathbb{C}$ is the reciprocal channel between the antennas \cite{nissel2022correctly}. To calibrate, each device transmits a pilot to the AP, which measures $g_{k\rightarrow \text{AP}}$, ignoring noise. The AP responds with a precoded DL pilot to each device
\begin{equation}\label{eq:calibration_coeff_meas}
	y_k = \frac{g_{\text{AP}\rightarrow k}}{g_{k\rightarrow \text{AP}}} = \frac{r_k}{t_k}\frac{t_\text{AP}}{r_\text{AP}}\equiv \frac{1}{c_k}.
\end{equation}
Then, with a single broadcast pilot from the AP over a new channel instance, say $h_k'\neq h_k$, the devices can estimate the UL $g_{k\rightarrow \text{AP}}'$ by multiplying the DL with the calibration coefficient $c_k$:
\begin{equation}\label{eq:reciprocity_channel_estimation}
	c_kg_{\text{AP}\rightarrow k}' = h_k't_k r_\text{AP} = g_{k\rightarrow \text{AP}}'.
\end{equation}

\subsection{Phase Drift from Wiener Phase Noise}
Relative to the amplitudes $|t_k|, |r_k|$ and the AP hardware $t_\text{AP}, r_\text{AP}$, the device hardware phases $\angle t_k, \angle r_k$ will drift significantly faster with time \cite{nissel2022correctly}. This drift is primarily caused by phase noise, which can be modeled as a Wiener process \cite{piemontese2024discrete} of Gaussian increments to the transmit chain $\angle t_k' = \angle t_k + \epsilon_k/2$, $\epsilon_k\sim\mathcal{N}(0,\alpha)$ and receive chain which shares the same oscillator $\angle r_k' = \angle r_k - \epsilon_k/2$ \cite{nissel2022correctly}, resulting in a rotated channel estimate 
\begin{equation}\label{eq:reciprocity_channel_estimation_phase_error}
	\widehat{g_{k\rightarrow \text{AP}}'} = g_{k\rightarrow \text{AP}}'e^{j\epsilon_k}.
\end{equation}
If a time $T$ has passed since calibration the increments will accumulate to a total phase offset distributed as $\mathcal{N}(0,T\alpha)$ such that (\ref{eq:reciprocity_channel_estimation_phase_error}) becomes increasingly uncertain with time. If $\alpha$ is sufficiently low, depending on the oscillator quality, calibration according to (\ref{eq:calibration_coeff_meas}, \ref{eq:reciprocity_channel_estimation}) can be done seldom ($T\gg 1$), enabling efficient estimation of the UL channel.
\section{Hybrid Channel Estimation}
The hybrid channel estimation scheme is split into two variants, A and B, summarized by Figures \ref{fig:tikz_variant_A} and \ref{fig:tikz_variant_B} respectively. Figure \ref{fig:tikz_estimation} summarizes the differences between A and B. The two variants are described in detail in the following sections.
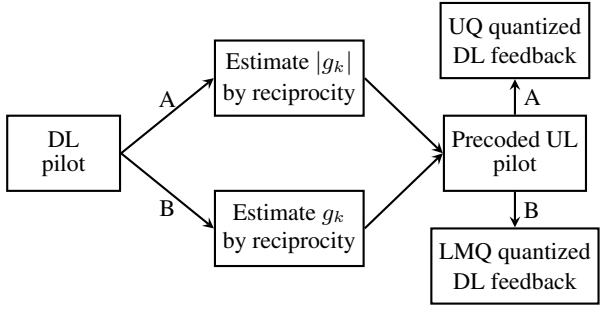
\begin{figure}[t!]
	\centering
	\begin{tikzpicture}[>=stealth, thick]
		
		\node[draw, rectangle, minimum width=1.5cm, minimum height=1.0cm] (A) at (0,0) {\shortstack{DL\\pilot}};
		

		\node[draw, rectangle, minimum width=1.5cm, minimum height=1.0cm] (B) at (3,1) {\shortstack{Estimate $|g_k|$\\by reciprocity}};
		\node[draw, rectangle, minimum width=1.5cm, minimum height=1.0cm] (C) at (3,-1) {\shortstack{Estimate $g_k$\\by reciprocity}};
		
		\node[draw, rectangle, minimum width=1.5cm, minimum height=1.0cm] (D) at (6,0) {\shortstack{Precoded UL\\pilot}};
		
		\node[draw, rectangle, minimum width=1.5cm, minimum height=1.0cm] (E) at (6.00,1.5) {\shortstack{UQ quantized\\DL feedback}};
		
		\node[draw, rectangle, minimum width=1.5cm, minimum height=1.0cm] (F) at (6.00,-1.5) {\shortstack{LMQ quantized\\DL feedback}};
		

		\draw[->] (A.east) -- (B.west) node[midway, above]  {A};
		\draw[->] (A.east) -- (C.west) node[midway, below]  {B};
		
		\draw[->] (B.east) -- (D.west) node[midway, above]  {};
		\draw[->] (C.east) -- (D.west) node[midway, above]  {};
		
		\draw[->] (D.north) -- (E.south) node[midway, right]  {A};
		\draw[->] (D.south) -- (F.north) node[midway, right]  {B};
		
		
		
	\end{tikzpicture}
	\caption{Channel estimation corresponding to steps 2,3,4 in Variant A and B.}
	\label{fig:tikz_estimation}
\end{figure}

\subsection{Variant A: Feedback-based phase estimation}
In Variant A, the devices estimate the UL channel amplitude from a DL pilot, but the phase from pilot-based quantized feedback. A detailed description follows:
\begin{enumerate}
	\item Devices and AP calibrate the hardware in amplitude but not phase (estimate $|c_k|$ but not $\angle c_k$), such that only $|g_k|$ can be estimated from reciprocity. As the amplitude is much more stable than the oscillator phase, this can be done rarely and to a high precision \cite{nissel2022correctly}. 
	\item AP broadcasts a DL pilot simultaneously to all users, from which devices estimate $|g_k|$ exactly. The devices now hold channel estimates ${\widetilde{g_k} \equiv |g_k|}$. 
	\item Devices transmit orthogonal UL pilots precoded with $\widetilde{g_k}$, from which the AP estimates the channel phase error $e_k\equiv\phi_k$ exactly.
	\item AP quantizes $e_k$ uniformly with (\ref{eq:UQ_levels}) and transmits $Q(e_k)$ error-free to each device. Devices update their channel estimate to the final $\widehat{g_k} = |g_k|e^{jQ(\phi_k)}$.
\end{enumerate}

\subsection{Variant B: Reciprocity-based phase estimation}
In Variant B, the devices estimate the UL amplitude \textit{and} the phase from the DL pilot, the pilot-based feedback can then be optimized according to the statistics of the phase noise. Compared to Variant A which is unaffected by phase noise, the performance of Variant B will depend on the amount of time $t\geq 0$ since calibration. As such, let $g_k^t$ with amplitude $|g_k^t|$ and phase $\phi_k^t$ be an i.i.d. instance of the channel in (\ref{eq:receive_OAC}) at time $t$. A detailed description follows: 
\begin{enumerate}
	\item Devices and AP calibrate for joint reciprocity by estimating $c_k$ at time $t=0$, as described in Section \ref{sec:reciprocity_calibration}. 
	\item AP broadcasts a DL pilot simultaneously to all users, from which devices estimate $g_k^t$ using (\ref{eq:reciprocity_channel_estimation}), giving an exact estimate of $|g_k^t|$, but an inexact phase estimate $\widetilde{\phi_k^{t}}$ as follows
	\begin{equation}
		\begin{split}
			&\widetilde{\phi_k^{t}} = \phi_k^t + e_k^t + \underbrace{\sum_{\tau=1}^{t-1}\left(e_k^\tau - Q(e_k^\tau)\right)}_{\equiv E^{t-1}},\\
			&e_k^t\sim\mathcal{N}(0,\alpha), \text{i.i.d. across $t,k$},
		\end{split}
	\end{equation}
	where $E^{t-1}$ is the cumulative quantization error from previous estimation rounds, unknown to each device. Because of phase noise, $\widetilde{\phi_k^{t}}$ is likely too inaccurate for coherent aggregation, but constitutes an intermediate  estimate $\widetilde{g_k}^t \equiv |g_k^t|e^{j\widetilde{\phi_k^{t}}}$.   
	\item Devices transmit orthogonal UL pilots precoded by $\widetilde{g_k}^t$, to which the AP observes
	\begin{equation}\label{eq:phase_error_B}
		\phi_k^t - \widetilde{\phi_k^t} = -e_k^t - E^{t-1},
	\end{equation}
	where compared to the devices, the AP knows $E^{t-1}$ which can be subtracted from (\ref{eq:phase_error_B}).
	\item As $e_k^t$ is Gaussian, it may be quantized optimally using a LMQ quantizer giving $Q(e_k^t)$, which is transmitted error-free to each device. The final estimate of $g_k^t$ at each device is ${\widehat{g_k}^t = |g_k^t|e^{j\widehat{\phi_k^t}}}$, where
	\begin{equation}
		\widehat{\phi_k^t} = \widetilde{\phi_k^t} - Q(e_k^t) = \phi_k^t +  E^{t}.
	\end{equation}
	\item The system increments $t$ by 1. If $t<T$ the system continues from step 2, otherwise the system re-calibrates by starting over from step 1, nulling $E^{t}$.
\end{enumerate}
The distribution of each element in $E^{t}$ is a mixture of translated and truncated Gaussians \cite{chopin2011fast}, 
making it difficult to characterize the distribution of $\text{exp}\left(j\left(\phi_k^t - \widehat{\phi_k^t}\right)\right)$ in closed form, which we have to leave to future work.

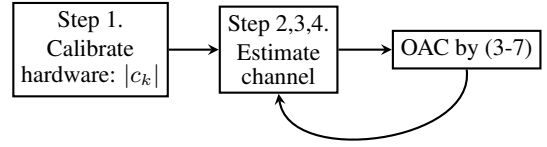
\begin{figure}[t!]
	\centering
	\begin{tikzpicture}[>=stealth, thick]
		\node[draw, rectangle, minimum width=1.5cm, minimum height=1.0cm] (A) at (0,0) {\shortstack{Step 1.\\ Calibrate\\hardware: $|c_k|$}};
		
		\node[draw, rectangle, minimum width=1.5cm, minimum height=1.0cm] (B) at (2.5,0) {\shortstack{Step 2,3,4.\\ Estimate\\channel}};
		
		\node[draw, rectangle, minimum width=1.5cm, minimum height=0.5cm] (C) at (5,0) {OAC by (\ref{eq:receive_OAC}-\ref{eq:OAC_estimate})};
		
		\draw[->] (A.east) -- (B.west) node[midway, above] {};
		
		\draw[->] (B.east) -- (C.west);
		
		\draw[->, bend left=90] (C.south) to (B.south);
		
	\end{tikzpicture}
	\caption{Variant A: As $|c_k|$ is very stable, the calibration is done only once, making Variant A independent of $t$. Steps 2,3,4 are summarized by Figure \ref{fig:tikz_estimation}.}
	\label{fig:tikz_variant_A}
\end{figure}
\begin{figure}[t!]
	\centering
	\begin{tikzpicture}[>=stealth, thick]
		\node[draw, rectangle, minimum width=1.5cm, minimum height=1.0cm] (A) at (0,0) {\shortstack{Step 1.\\ Calibrate\\hardware: $c_k$}};
		
		\node[draw, rectangle, minimum width=1.5cm, minimum height=1.0cm] (B) at (2.5,0) {\shortstack{Step 2,3,4.\\ Estimate\\channel}};
		
		\node[draw, rectangle, minimum width=1.5cm, minimum height=0.5cm] (C) at (5,0) {OAC by (\ref{eq:receive_OAC}-\ref{eq:OAC_estimate})};
		
		\node[draw, rectangle, minimum width=1.5cm, minimum height=0.5cm] (D) at (2.5,-2.0) {\shortstack{Step 5.\\ $t\geq T$?}};
		
		\draw[->] (A.east) -- (B.west) node[midway, above] {$t+1$};
		
		\draw[->] (B.east) -- (C.west);
		
		
		\draw[->] (D.north) -- (B.south) node[midway, right] {No};
		
		\draw[->] (C.south) |- (D.east) node[midway, above=5pt, right=0pt, left=5pt] {$t+1$};
		
		\draw[->] (D.west) -| (A.south) node[midway, above=5pt, right=10pt] {Yes};
		
		\draw[->] (D.west) -| (A.south) node[midway, below=5pt, right=10pt] {$t=0$};
		
	\end{tikzpicture}
	\caption{Variant B: The calibration is done periodically because of oscillator drift, making $t$ important. Steps 2,3,4 are summarized by Figure \ref{fig:tikz_estimation}.}
	\label{fig:tikz_variant_B}
\end{figure}
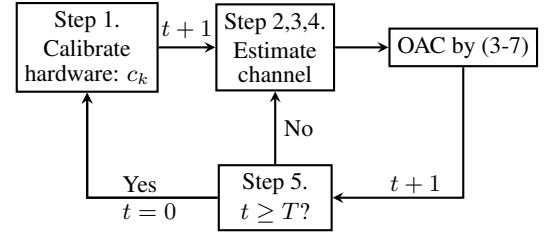

\section{Analytical result for Variant A}
\begin{lemma}[Variant A MSE]
	Let $\widehat{v}$ be given by OAC with Variant A channel estimation, where $n, g_k\sim\mathcal{CN}(0,1), \text{ i.i.d. } \forall k$, then
	\begin{equation}
		\mathbb{E}\left[|\widehat{v}-v|^2\right] = 2K\left(1 - \frac{2^N}{\pi}\text{sin}\left(\frac{\pi}{2^N}\right)\right) + 1.
	\end{equation}
\end{lemma}

\begin{proof}
	\begin{equation}
		\mathbb{E}\left[|\widehat{v}-v|^2\right] = \mathbb{E}\left[(\widehat{v}-v)^*(\widehat{v}-v)\right] = \sum_{k=1}^{K}\mathbb{E}\left[|e_k|^2\right] + 1,
	\end{equation}
	where $e_k = e^{j(\phi_k-Q(\phi_k))}- 1$. The distribution of $\phi_k-Q(\phi_k)$ is given by
	\begin{equation}
		\phi_k-Q(\phi_k)\sim\text{U}\left(-\frac{\pi}{2^N},\frac{\pi}{2^N}\right),
	\end{equation}
	which follows from $\phi_k\sim\text{U}(0, 2\pi)$ and that the quantized feedback constrains $\phi_k$ to an increasingly halved unit-circle. The mean of $e^{j(\phi_k-Q(\phi_k))}$, denoted by $\widetilde{\mathbb{E}}$, is computed by first principles:
	
	\begin{equation}
		\widetilde{\mathbb{E}} = \int_{-\frac{\pi}{2^N}}^{\frac{\pi}{2^N}}\frac{2^N}{2\pi}e^{jx} dx = \frac{2^N}{\pi}\text{sin}\left(\frac{\pi}{2^N}\right).  
	\end{equation}
	Then one has
	\begin{equation}
		\begin{split}
			&\mathbb{E}\left[|e_k|^2\right] = 2 - 2\widetilde{\mathbb{E}} = 2 - \frac{2^{N+1}}{\pi}\text{sin}\left(\frac{\pi}{2^N}\right),
		\end{split}
	\end{equation}
	and finally
	\begin{equation}
		\mathbb{E}\left[|\widehat{v}-v|^2\right] = 2K\left(1 - \frac{2^N}{\pi}\text{sin}\left(\frac{\pi}{2^N}\right)\right) + 1.
	\end{equation}
\end{proof}
\begin{figure}[t!]
	\centering
	\includegraphics[height=2.2in]{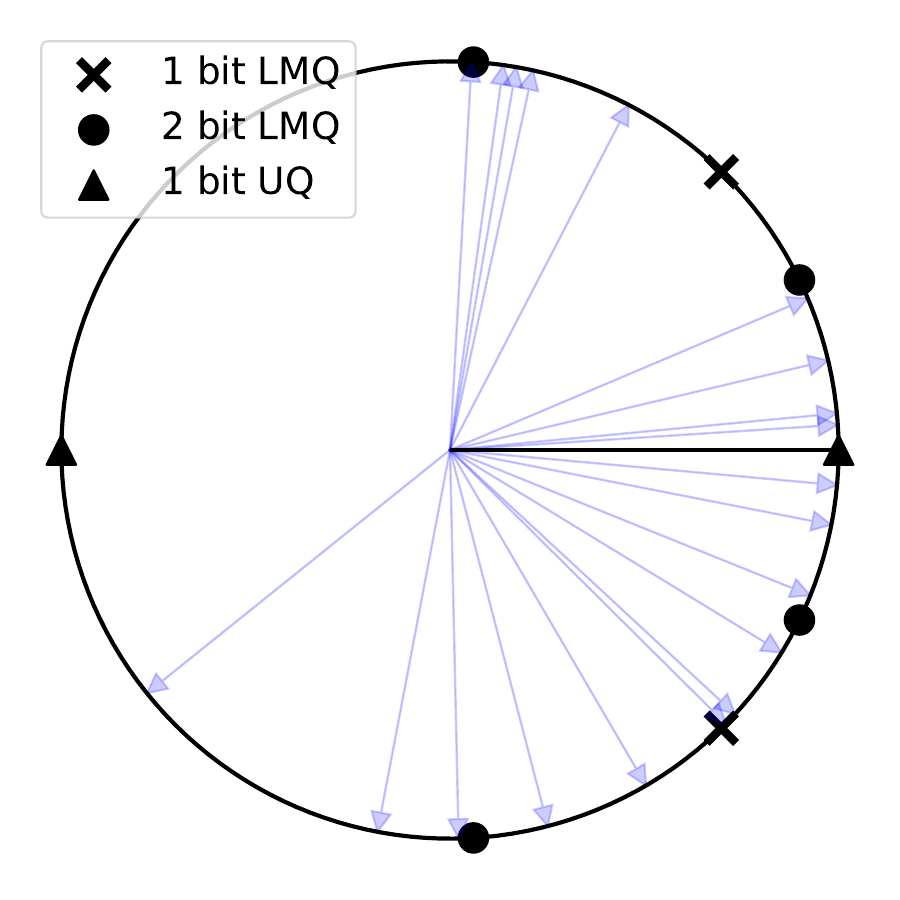}
	\caption{Lloyd-Max quantizer (LMQ) vs uniform quantizer (UQ) with 20 instances of angles sampled i.i.d. from $\mathcal{N}(0,1)$.}\label{fig:quantizers_unit_circle}
\end{figure}
\section{Simulation results}
In a first simulation, presented in Figure \ref{fig:first_simulation}, the proposed variants A and B are simulated for a varying number $N$ of quantization bits. The results show that if $\alpha$ is sufficiently low, which follows with a stable oscillator, Variant B achieves a lower MSE than A even when $T$ is high, however, the difference is negligible for  high $N$. Furthermore, even for low $\alpha$, the quantized phase feedback decreases the MSE significantly. Note that the reciprocity calibration for Variant B requires no additional signaling compared to Variant A, making the comparison fair in terms of resource usage. In the second simulation, presented in Figure \ref{fig:second_simulation}, the MSE for Variant B when $T$ (The number of OAC iterations before re-calibration) grows is approximated. The MSE grows approximately linearly with $T$ for low $\alpha$, but for high $\alpha$ and small $N$ there is a clear non-linear growth of the MSE. This phenomena could be shown analytically by approximating the quantization errors in $E^t$ as uniformly distributed when $\alpha$ is low or $N$ high. All relevant parameters are specified in the figure captions and legends. 
\begin{figure}[t!]
	\centering
	\includegraphics[height=2.2in]{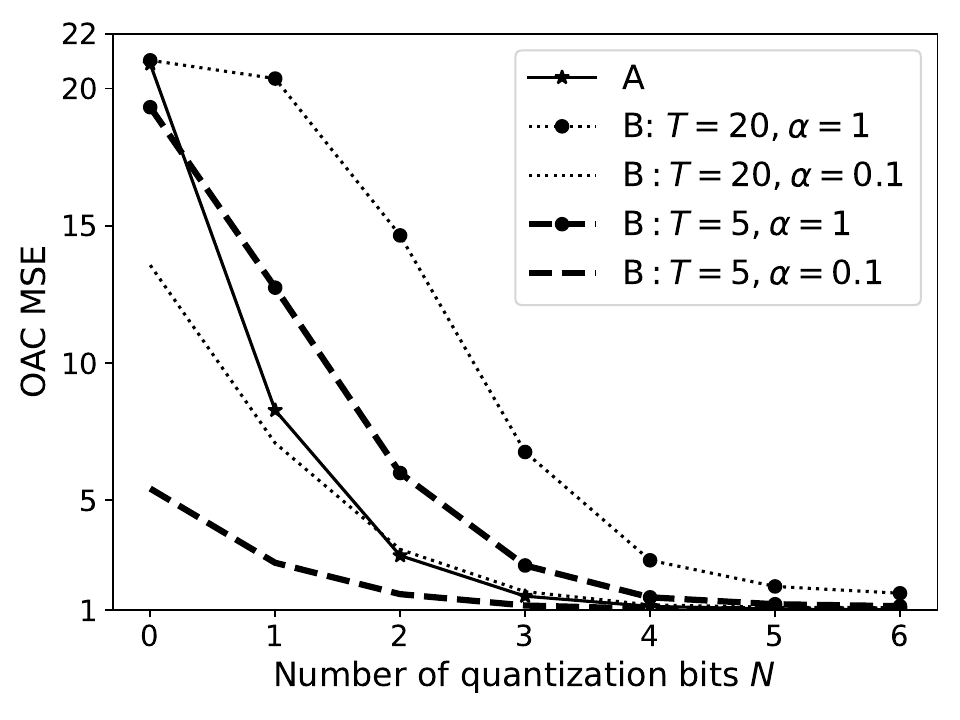}
	\caption{OAC sum estimator (\ref{eq:OAC_estimate}) MSE (\ref{eq:MSE_goal}) against number of quantization bits $N$.  $N=0$ means no phase feedback. $K=10$.  $v_k, g_k, n \sim\mathcal{CN}(0,1)$ i.i.d. $\forall k$, averaged over $10^5$ trials.}\label{fig:first_simulation}
\end{figure}
\begin{figure}[t!]
	\centering
	\includegraphics[height=2.2in]{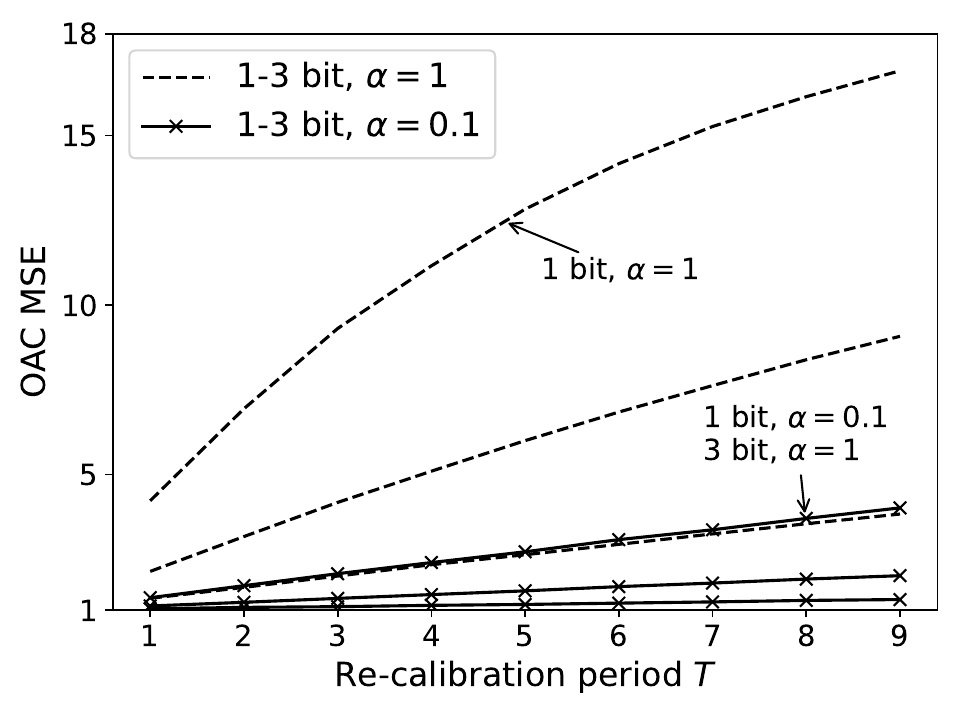}
	\caption{OAC sum estimator (\ref{eq:OAC_estimate}) MSE (\ref{eq:MSE_goal}) against the re-calibration period $T$ of the phase. $K=10$,  $v_k, g_k, n \sim\mathcal{CN}(0,1)$ i.i.d. $\forall k$, averaged over $10^5$ trials.}\label{fig:second_simulation}
\end{figure}
\label{sec:typestyle}

\section{conclusions}
This work proposes a hybrid channel estimation scheme, aiming to reduce the signaling overhead for over-the-air computation. The scheme allows selecting the precision of channel amplitude and phase estimation separately by relying on both calibrated reciprocity-based and feedback-based channel estimation. If phase noise is low, one may estimate both channel amplitude and phase from reciprocity, otherwise channel phase feedback is necessary. Additionally, quantization may be optimized to the phase noise distribution. A future direction is to derive an approximate MSE for the second variant of the proposed scheme.

\bibliographystyle{IEEEbib}
\bibliography{references}

\end{document}